\documentclass{article}

\usepackage{arxiv}
\usepackage{cite}

\usepackage{color}
\usepackage{cite}
\usepackage{amsmath,amssymb,amsfonts,amsthm,bm,bbold}
\usepackage{mathtools}
\usepackage{algorithmic}
\usepackage{graphicx}
\usepackage{textcomp}
\usepackage{dsfont}
\usepackage{epstopdf}
\usepackage{orcidlink}
\usepackage{enumerate}
\usepackage{lscape}
\usepackage{multicol}
\usepackage{multirow}
\usepackage{relsize}
\usepackage{booktabs}
\usepackage{framed}  
\usepackage{empheq}
\usepackage{graphics} 
\usepackage{epsfig} 
\usepackage{times} 
\usepackage{graphicx}
\usepackage{comment} 
\usepackage{apptools}
\usepackage{theoremref}
\usepackage{balance}
\usepackage{tensor}
\usepackage{balance}
\usepackage{hyperref}
\usepackage{float}
\usepackage{subcaption}
\graphicspath{{./},{./figures/}}

\hypersetup{
    colorlinks=true,
    linkcolor=blue,
    filecolor=magenta,      
    urlcolor=cyan,
    citecolor=blue,
    pdfpagemode=FullScreen,
    }

\newtheorem{prop}{Proposition}

\newcommand{\vsp}{\vspace*{5pt}}

\date{}

\begin{document}

\title{\huge 
Regulation of a continuously monitored quantum harmonic oscillator with inefficient detectors
}
\author{Ralph Sabbagh\orcidlink{https://orcid.org/0000-0002-2020-5420}, 
 Olga Movilla Miangolarra\orcidlink{https://orcid.org/0000-0002-9214-8525}, and Tryphon T. Georgiou\orcidlink{https://orcid.org/0000-0003-0012-5447}, 
\thanks{This work was supported in part by the AFOSR under
grant FA9550-24-1-0278, the ARO under W911NF-22-1-0292, and the NSF under ECCS-2347357. OMM acknowledges support by the European Union's Horizon 2020
programme under the MSCA G.A.~No.\ 101151140.}
\thanks{R. Sabbagh and T. T. Georgiou are with the Department of Mechanical and Aerospace
Engineering, University of California, Irvine, CA 92697 USA 
        {\tt\small rsabbag1@uci.edu, tryphon@uci.edu}}
       \thanks{O. Movilla Miangolarra is with the Department of Physics, Universidad de La Laguna, La Laguna 38203, Spain, and Instituto Universitario de Estudios Avanzados (IUdEA), Universidad de La Laguna, La Laguna 38203, Spain 
        {\tt\small omovilla@ull.edu.es}}}

\maketitle

\begin{abstract}
We study the control problem of regulating the purity of a quantum harmonic oscillator in a Gaussian state via weak measurements. Specifically, we assume time-invariant Hamiltonian dynamics and that control is exerted via the back-action induced from monitoring the oscillator's position and momentum observables; the manipulation of the detector measurement strengths regulates the purity of the target Gaussian quantum state. After briefly drawing connections between Gaussian quantum dynamics and stochastic control, we focus on the effect of inefficient detectors and derive closed-form expressions for the transient and steady-state dynamics of the state covariance. We highlight the degradation of attainable purity that is due to inefficient detectors, as compared to that dictated by the Robertson-Schr\"odinger uncertainty relation. Our results suggest that quantum correlations can enhance the purity at steady-state.  The quantum harmonic oscillator represents a basic system where analytic formulae may provide insights into the role of inefficient measurements in quantum control; the gained insights are pertinent to measurement-based quantum engines and cooling experiments.
\end{abstract}


\section{INTRODUCTION}
The transformative progress of recent years\footnote{The 2012 Nobel prize to Serge Haroche and David Wineland was awarded for ``measuring and manipulation of individual quantum systems.''} in our ability to measure and manipulate quantum states has highlighted the role that control theory can play in the on-going quantum revolution \cite{wiseman2009quantum,jacobs2014quantum,dalessandro}.
Qualitative new features arise in quantum physics when measurements, that are no longer projective, allow monitoring conjugate observables simultaneously, and when a continuous sequence of such weak measurements renders the evolution of the quantum state into that of a stochastic process. Such continuous monitoring opens up the possibility of feedback control \cite{PhysRevA.88.042110}, by dynamically regulating the system Hamiltonian and monitoring parameters.

\begin{figure}[!t]
    \centering
    \includegraphics[width=.85\linewidth]{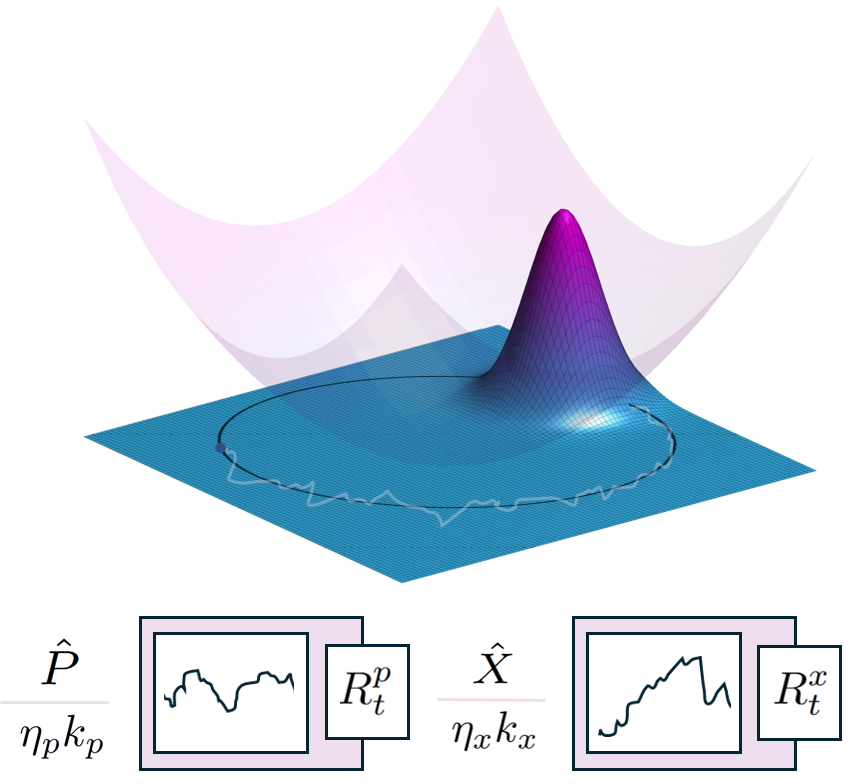}
    \caption{Schematic of the Wigner distribution of a quantum harmonic oscillator Gaussian state subject to the simultaneous continuous monitoring of its position and momentum observables $\hat{X}$ and $\hat{P}$, respectively. The system is coupled to two detectors monitoring the system with strength $k$ and efficiency $\eta$. The Gaussian packet is confined in a harmonic potential and rotates about the trajectory of the unitary dynamics. The roughness of the paths is a direct manifestation of the stochastic back-action caused by continually probing the system via an apparatus that provides measurement readouts $R_t^x$ and $R_t^p$. Such measurement schemes have been successfully implemented on qubits \cite{hacohen2016quantum}.}
    \label{fig:enter-label}
    \vspace{-6mm}
\end{figure}

In this letter, we consider a quantum harmonic oscillator and the control of Gaussian states that are subject to continuous monitoring. Such a system is sufficiently simple that can be dealt with analytically, and still the topic of continued interest \cite{PhysRevLett.123.163601} as it plays
a fundamental role in many areas of quantum physics, including quantum optics and quantum information science.
Moreover, continuous Gaussian measurements are some of the most widely used in quantum labs \cite{landi}.

The present work builds on \cite{karmakar2022stochastic}, where the system's stochastic
and average trajectories for constant measurement strengths and fully efficient detectors were studied; 
the authors in \cite{karmakar2022stochastic} adopt a path-integral approach
using the formalism in \cite{PhysRevA.88.042110}. 
In this letter, we extend some of the analysis in \cite{karmakar2022stochastic} to the case of inefficient measurement detectors.

The letter is structured as follows. In Section \ref{sec:II}, we discuss the quantum harmonic oscillator, introduce the concept of state purity, and present the model of continuous measurement of Gaussian states. Section \ref{sec:III} discusses the dynamics of the state-covariance, while Section \ref{sec:IV} investigates the effect of the detector parameters on the purity. Our results suggest that quantum correlations can enhance the purity at steady-state.

\section{Background and stochastic model}\label{sec:II}
 The quantum harmonic oscillator is introduced next, along with the notions of purity and of Gaussian states.
 Then, a stochastic model describing the conditional evolution of the system under continuous monitoring of its position and momentum is presented.
 Finally, the stochastic dynamics are specialized to Gaussian states, which will be the setting of this letter. The Gaussian assumption reduces the general dynamics to a set of equations describing the evolution of the mean and covariance of the quantum system.  Common states, including coherent states, and measurement noise, are Gaussian which makes this assumption both reasonable and useful \cite{jacobs2006straightforward,genoni2016conditional,PhysRevA.96.062131,WANG20071}. 
 
\subsection{The quantum harmonic oscillator }
The one-dimensional quantum harmonic oscillator \cite[Appendix A.4]{wiseman2009quantum} is described by the energy Hamiltonian
\begin{align*}
    \hat{H}=\cfrac{1}{2m}\hat{P}+\cfrac{1}{2}m\omega^2\hat{X}^2,
\end{align*}
where $m$ and $\omega$ denote the mass and frequency of the oscillator, respectively, and $\hat{X}$ and $\hat{P}$ denote the position and momentum operators on a Hilbert space $\mathbb{H}$, respectively. These operators satisfy the canonical commutation relation $[\hat{X},\hat{P}]=i\hbar \hat{I}$, where $[\cdot\,,\cdot]$ denotes the commutator operation, $\hbar$ the reduced Planck's constant, and $\hat{I}$ the identity operator.

\subsection{Density operators, purity, and Gaussian states}
The state of a quantum harmonic oscillator  can be described by a density operator $\rho$ on a Hilbert space $\mathbb{H}$, that is, a linear operator in $\mathbb{H}$ satisfying
\begin{align*}
    \text{tr}(\rho)=1,~~~\rho^{\dagger}=\rho,~~~\rho\geq 0,
\end{align*}
where $\text{tr}$ and $(\cdot)^{\dagger}$ denote the trace and adjoint operations on $\mathbb{H}$, respectively. An equivalent representation of the quantum state $\rho$ can be obtained as a bivariate function in phase space  $W_{\rho}(\mathbf{s})$, $\mathbf{s}\in\mathbb{R}^2$, known as the \textit{Wigner distribution of $\rho$}. This representation is readily obtained by taking the scaled Fourier transform of the characteristic function of $\rho$,
\begin{align*}
    f_{\rho}(\text{\boldmath$\xi$})=\text{tr}\left(\rho e^{i\xi_1\hat{X}+i\xi_2\hat{P}}\right),~~
\end{align*}
where $\text{\boldmath$\xi$}=(\xi_1,\xi_2)^\top\in\mathbb{R}^2$, that is,
\begin{align*}
    W_{\rho}(\mathbf{s})=\cfrac{1}{(2\pi)^2}\int_{\mathbb{R}^2}f_{\rho}(\text{\boldmath$\xi$})e^{-i\mathbf{s}^\top\text{\boldmath$\xi$}}\text{d}\text{\boldmath$\xi$}.
\end{align*}
An important figure in physical experiments is the \textit{purity} $p$ of $\rho$,
 defined as\footnote{The second equality can be established using the Weyl-Wigner transforms.} 
\begin{align*}
    p:= \text{tr}(\rho^2) = (2\pi\hbar)\int_{\mathbb{R}^2}W^2_{\rho}(\mathbf{s})\text{d}\mathbf{s},
\end{align*}
where $0< p\leq 1$. A state $\rho$ is said to be pure when $p=1$, and mixed otherwise. 

Of great importance are  Gaussian states; we say that $\rho$ is  \textit{Gaussian} if its Wigner distribution is given by a bivariate Gaussian function \cite{WANG20071}
\begin{align*}
    W_{\rho}(\mathbf{s})= \cfrac{1}{2\pi\sqrt{\det(\Sigma)}}e^{-\frac{1}{2}(\mathbf{s}-\text{\boldmath$\mu$})^\top\Sigma^{-1}(\mathbf{s}-\text{\boldmath$\mu$})},
\end{align*}
for some mean $\text{\boldmath$\mu$} = (\mu^{x},~\mu^p)^\top\in\mathbb{R}^2$ and $2\times 2$ positive-definite covariance matrix \cite{PhysRevA.49.1567}
\begin{align*}
    \Sigma = \begin{bmatrix}
        v^x&c\phantom{^x}\\
        c\phantom{^x}&v^p
    \end{bmatrix} > 0.
\end{align*}
The entries of $\text{\boldmath$\mu$}$ and $\Sigma$, which correspond to the expected values, variance, and symmetrized covariance of $\hat{X}$ and $\hat{P}$ at $\rho$, can be computed as follows
\begin{align}
    &\mu^x=\text{tr}(\rho\hat{X}),~~\mu^p=\text{tr}(\rho\hat{P}),\label{eq:means}\\
    &v^x=\text{tr}(\rho\hat{X}^2)- \text{tr}(\rho\hat{X})^2,~~v^p=\text{tr}(\rho\hat{P}^2)- \text{tr}(\rho\hat{P})^2,\label{eq:variances}\\
    &c = \text{tr}(\rho(\hat{X}\hat{P}+\hat{P}\hat{X}))/2-\text{tr}(\rho\hat{X})\text{tr}(\rho\hat{P})\label{eq:covariance}.
\end{align}
Thus, Gaussian states are completely characterized by their mean \text{\boldmath$\mu$} and covariance $\Sigma$. 
\subsection{Stochastic master equation and monitoring of position and momentum}
Continuous quantum measurement is the time-continuum limit of a sequence of weak measurements, whose strengths scale with the time duration over which they are performed \cite{jacobs2006straightforward}. Namely, the quantum state $\rho_t$ channels, in a small time interval $\Delta t$, through the maps 
\begin{align*}
    \rho_t\mapsto \rho_{t+\Delta t/2} &= \cfrac{e^{-\frac{k_x\Delta t}{8}(\hat{X}-r_x)^2}\rho_te^{-\frac{k_x\Delta t}{8}(\hat{X}-r_x)^2}}{\text{tr}\left(e^{-\frac{k_x\Delta t}{8}(\hat{X}-r_x)^2}\rho_te^{-\frac{k_x\Delta t}{8}(\hat{X}-r_x)^2}\right)},\\
    \rho_{t+\Delta t/2}\mapsto \rho_{t+\Delta t} &= \cfrac{e^{-\frac{k_p\Delta t}{8}(\hat{P}-r_p)^2}\rho_te^{-\frac{k_p\Delta t}{8}(\hat{P}-r_p)^2}}{\text{tr}\left(e^{-\frac{k_p\Delta t}{8}(\hat{P}-r_p)^2}\rho_te^{-\frac{k_p\Delta t}{8}(\hat{P}-r_p)^2}\right)},
\end{align*}
where $r_x$ and $r_p$ are the values obtained from weakly measuring $\hat{X}$, then $\hat{P}$, respectively. The parameters $k_x>0$ and $k_p>0$ characterize the strength of the measurements being performed, and will constitute our control parameters in this letter. By repeatedly applying the above quantum channels and taking the limit as $\Delta t\rightarrow d t$, one obtains the conditional evolution equation of the harmonic oscillator state $\rho_t,~t\geq 0$ \cite{jacobs2006straightforward}. Combining this with the Hamiltonian evolution, we get
    \begin{align}
        &\text{d}\rho_t = \cfrac{1}{i\hbar}[\hat{H},\rho_t]\text{d}t -\frac{k_x}{8}[\hat{X},[\hat{X},\rho_t]]\text{d}t-\frac{k_p}{8}[\hat{P},[\hat{P},\rho_t]]\text{d}t \nonumber\\
        &+ \frac{1}{2}\sqrt{\eta_xk_x}(\hat{X}\rho_t + \rho_t \hat{X} - 2\langle \hat{X}\rangle\rho_t)\text{d}W^x_t\nonumber\\
        & + \frac{1}{2}\sqrt{\eta_pk_p}(\hat{P}\rho_t + \rho_t \hat{P} - 2\langle \hat{P}\rangle\rho_t)\text{d}W^p_t,\label{eq:Belavkin}
    \end{align}
where $W_t^x$ and $W_t^p$ are independent standard Brownian noises obtained from the integrated measurement readouts $R^x$ and $R^p$ of position and momentum, respectively (see \eqref{eq:records}). The effect of the terms besides the one containing the Hamiltonian on the evolution of $\rho_t$ is what is known as the \textit{measurement back-action}. The parameters $0< \eta_x,\eta_p\leq 1$ reflect the detector efficiencies, where $1$ corresponds to ideal detection and $0$ corresponds to no detection. Measurements are typically realized by coupling the system to auxiliary
systems (\textit{meters}), so that their interaction effects a measurement back-action \cite{jackson2023perform}. Equation \eqref{eq:Belavkin} is referred to as the \textit{Stochastic Master Equation} (SME), originally derived by V. Belavkin \cite{blaquiere1987information}.

The integrated measurement readouts $R^x$ and $R^p$\footnote{In the continuum limit, one should think of the quantities  $\text{d}R_t^x/\text{d}t\approx r_x$ and $\text{d}R_t^p/\text{d}t \approx r_p$ as being the weakly measured values at each instant $t$.} are modeled by the stochastic differential equations
\begin{subequations}\label{eq:records}
\begin{align}
     \text{d}R_t^x &= \mu_t^x\text{d}t+\frac{1}{\sqrt{\eta_xk_x}}\text{d}W_t^x,\\
    \text{d}R_t^p &= \mu_t^p\text{d}t+\frac{1}{\sqrt{\eta_pk_p}}\text{d}W_t^p.
\end{align}
  \end{subequations} 
Notice how the Brownian noises $W_t^x$ and $W_t^p$ in \eqref{eq:records} are the exact same as the ones appearing in \eqref{eq:Belavkin}. This illustrates the explicit dependence of the state evolution $\rho_t$ on the speciﬁc measurement results $R_t^x$ and $R_t^p$. The parallels between (\ref{eq:Belavkin}-\ref{eq:records}) and nonlinear filtering are inescapable \cite{altafini2012modeling}.

A schematic of the Wigner distribution of a monitored Gaussian state is shown in Fig. \ref{fig:enter-label}. The system is coupled to two detectors, each monitoring the
system with a strength $k$ and efficiency $\eta$. The Gaussian packet is confined in a harmonic potential and rotates about the trajectory of the unitary dynamics (i.e., the ones induced by the Hamiltonian term only). For a detailed derivation of the Master Equation and associated output equations, see \cite{jacobs2006straightforward}.

  If the initial state $\rho_0$ is Gaussian, then the evolution $\rho_t$ under \eqref{eq:Belavkin} remains Gaussian\footnote{Although the operator in \eqref{eq:Belavkin} is nonlinear, Gaussianity is preserved due to the fact that the nonlinear comes from a scaling factor in the noise term.} for all $t>0$ \cite{genoni2016conditional,PhysRevLett.122.190402,PhysRevA.96.062131}. Thus, it is enough to consider the dynamics of  $(\text{\boldmath$\mu$}_t,\Sigma_t)$, which can be derived using (\ref{eq:means}-\ref{eq:Belavkin}) and a careful application of It$\hat{\text{o}}$'s Lemma \cite{jacobs2006straightforward}. 
  These are
   \begin{align}
       \text{d}\text{\boldmath$\mu$}_t &= A^\top\text{\boldmath$\mu$}_t\text{d}t + \Sigma_t B\text{d}W_t,\label{eq:mean}\\
       \dot{\Sigma}_t &= \Sigma_t A+A^{\top}\Sigma_t-\Sigma_t BB^{\top}\Sigma_t+Q,\label{eq:cov}
   \end{align}
   where $Q = \hbar^2\text{diag}(k_p,k_x)/4$ and 
   \begin{align}
       A = \begin{bmatrix}
           0&-m\omega^2\\
           \frac{1}{m}&0
       \end{bmatrix},B=\begin{bmatrix}
           \eta_xk_x&\hspace{-4mm}0\\
           0&\eta_pk_p
       \end{bmatrix}^{\frac{1}{2}}.\label{eq:ABQ}
   \end{align}
  While the dynamics of $\text{\boldmath$\mu$}_t$ undergo noisy rotation in phase-space, those of $\Sigma$ are deterministic. In fact, \eqref{eq:cov} is nothing but a Riccati Differential Equation, whose convergence and analytical solution is discussed in the following section. 
   
   \section{Dynamics of the state-covariance}\label{sec:III}
Herein, 
we derive closed-form expressions for the steady-state covariance $\Sigma_\infty$ (that satisfies the Algebraic Riccati Equation \eqref{eq:ARE}), as well as provide an explicit expression for the transient response $\Sigma_t$ (for $t\in[0,\infty)$ that starts from an initial condition $\Sigma_0>\Sigma_\infty$. This condition is natural in our setting, since continuous measurement typically reduces uncertainty.

While the theory of the Riccati equation is a standard topic in classical textbooks, e.g., \cite{anderson2007optimal}, the specific form and size of matrices allows explicit expressions for $\Sigma_\infty$. Moreover, the expression \eqref{eq:convergencerate} of the transient is mildly original, and highlights in a rather transparent manner the convergence in our case of the solution of the Differential Riccati Equation to the stationary value.

\subsection{Stationary solution}
Since the pair $(A,B)$ is stabilizable (as, in fact, $B$ is already square and non-singular), it is well-known that the corresponding Algebraic Riccati Equation 
\begin{align}
      0=\Sigma A+A^{\top}\Sigma-\Sigma BB^{\top}\Sigma +Q,\label{eq:ARE}
\end{align}
has a unique positive definite solution $\Sigma_\infty$ \cite{bittanti1991riccati}, that the eigenvalues of
\begin{equation}\label{eq:Gamma}
\Gamma := A-BB^{\top}\Sigma_{\infty}
\end{equation}
have negative real parts, and that the Riccati Differential Equation converges to $\Sigma_{\infty}$ as $t\rightarrow \infty$, from any initial condition $\Sigma_0>0$ \cite{callier1995convergence}. Next, using the form and size of our data set \eqref{eq:ABQ}, we explicitly compute 
\begin{align*}
  \Sigma_{\infty}=: \begin{bmatrix}
       v_{\infty}^x&c_{\infty}\\
       c_{\infty}&v_{\infty}^p
   \end{bmatrix}.
\end{align*}

   \begin{prop}\label{prop:prop1}
   For detector efficiencies $\eta_x,\eta_p\in(0,1]$ and measurement strengths $k_x,k_p >0$, 
 \begin{align}
        v_{\infty}^x &= \cfrac{\gamma}{\eta_xk_x}\sqrt{1-\left(\cfrac{\alpha-\beta}{\gamma+\delta}\right)^2},~c_{\infty} = \cfrac{\alpha\gamma+\beta\delta}{(\gamma+\delta)\sqrt{\eta_x\eta_pk_xk_p}},\nonumber\\
        v_{\infty}^p &= \cfrac{\delta}{\eta_pk_p}\sqrt{1-\left(\cfrac{\alpha-\beta}{\gamma+\delta}\right)^2}\label{eq:explicit},
    \end{align}
  where
   \begin{align}
   \alpha &= -m\omega^2\sqrt{\frac{\eta_pk_p}{\eta_xk_x}},~~\beta= \frac{1}{m}\sqrt{\frac{\eta_xk_x}{\eta_pk_p}}, \nonumber\\
      \gamma &= \sqrt{\frac{\hbar^2}{4}k_pk_x\eta_x+
      \beta^2 
      },~\delta = \sqrt{\frac{\hbar^2}{4}k_xk_p\eta_p+
       \alpha^2 
       }.\label{eq:greek}\\\nonumber
   \end{align}
   \end{prop} 
   
   \begin{proof} We first rewrite the Algebraic Riccati Equation \eqref{eq:ARE} in the form
   \begin{align}
           &(\Sigma_\infty B-A^{\top}B^{-\top}) (B^\top\Sigma_\infty -B^{-1}A)\nonumber\\
           &=Q+ A^{\top}(BB^\top)^{-1}A,\label{eq:proof}
   \end{align}
   and then multiply both sides of \eqref{eq:proof} by $B$.
 We define
 $\tilde{\Sigma}_{\infty}:= B\Sigma_{\infty}B,$
       $\tilde{A}:=B^{-1}AB,$ and  $\Omega:=BQB$. It follows that 
   \begin{align}
       (\tilde{\Sigma}_{\infty}-\tilde{A}^{\top}) (\tilde{\Sigma}_{\infty}-\tilde{A}^{\top})^{\top}=\Omega+\tilde{A}^{\top}\tilde{A}.\label{eq:quad}
   \end{align}
   Thus, the left-hand factor in \eqref{eq:quad} must be given by 
   \begin{align}
       \tilde{\Sigma}_{\infty}-\tilde{A}^{\top} = (\Omega+\tilde{A}^{\top}\tilde{A})^{\frac{1}{2}}U,\label{eq:sigtilde}
   \end{align}
   for some orthogonal matrix $U$. In fact, this orthogonal matrix must be a rotation, i.e., of the form 
   \begin{align*}
      U= \begin{bmatrix}
         \phantom{-}\cos(\theta)&\sin(\theta)\\
         -\sin(\theta)&\cos(\theta)
       \end{bmatrix},
   \end{align*}
   since the diagonal elements must have the same sign.
   To see this, note that both $\tilde{\Sigma}_{\infty}$ and $\Omega+\tilde{A}^{\top}\tilde{A}$ are diagonal and that $\tilde{A}^{\top}$ is anti-diagonal, specifically,
   \begin{align*}
\tilde{A}=\begin{bmatrix}
          0&\alpha\\
           \beta&0
       \end{bmatrix},~(\Omega+\tilde{A}^{\top}\tilde{A})^{1/2}=\begin{bmatrix}
\gamma &0\\
0&\delta
       \end{bmatrix},
   \end{align*}
   with $\alpha, \beta,\gamma,\delta$ as in the statement of the proposition.
   In light of \eqref{eq:quad} and the symmetry of $\tilde\Sigma_\infty$, it must hold that  $\alpha -\delta\sin(\theta)=\beta+\gamma\sin(\theta)$. Therefore,
   \begin{align*}
       \sin(\theta) &=\cfrac{\alpha-\beta}{\gamma+\delta}~~\text{ and}~~\cos(\theta)= +\sqrt{1-\sin^2(\theta)},
   \end{align*}
where the positive value of the square root is chosen to ensure that $\tilde{\Sigma}_{\infty}$ is positive. Rearranging for $\Sigma_{\infty}$ in \eqref{eq:sigtilde} yields the steady-state covariance matrix  
   \begin{align}
  \Sigma_{\infty} = B^{-1} (\tilde{A}^{\top}+(\Omega+\tilde{A}^{\top}\tilde{A})^{1/2}U)B^{-1}.\label{eq:prop1}
   \end{align}
Further substituting expressions for the entries of $B$, $\tilde{A}$, $\Omega$, and $U$ yields the result.
    \end{proof}

\subsection{Transient response}
Knowing $\Sigma_{\infty}$, e.g., as obtained in \eqref{eq:explicit}, we derive a closed-form solution for the Riccati Equation \eqref{eq:cov}
that displays in a rather transparent form the transient dynamics, which appears to be mildly original. 
\begin{prop} For $\Sigma_0>\Sigma_\infty>0$, and  $\Gamma$ as in \eqref{eq:Gamma},
the solution to the Riccati Differential Equation \eqref{eq:cov} is given by 
\begin{align}\nonumber 
\Sigma_t=&\Sigma_\infty+e^{\Gamma^\top t}\left(
(\Sigma_0-\Sigma_\infty)^{-1} +\phantom{\int_0^t}\right.\\[-.1in]
&\phantom{xxxxxxxxx}\left. + \int_0^t e^{\Gamma \tau}BB^\top e^{\Gamma^\top \tau}dt
\right)^{-1}
e^{\Gamma t}.\label{eq:convergencerate}
\\\nonumber
\end{align}
\end{prop} 

\begin{proof}
We first subtract \eqref{eq:ARE} from  \eqref{eq:cov}, and use the expression for $\Gamma$, to obtain that
\begin{align}\nonumber
    \frac{\text{d}}{\text{d}t}(\Sigma_t-\Sigma_{\infty})&=(\Sigma_t-\Sigma_{\infty})\Gamma + \Gamma^{\top}(\Sigma_t-\Sigma_{\infty})\\
    &-(\Sigma_t-\Sigma_{\infty})BB^{\top}(\Sigma_t-\Sigma_{\infty}).\label{eq:Sigmaminusinfty}
\end{align}
We now define $Y_t=e^{-\Gamma^\top t}(\Sigma_t-\Sigma_{\infty})e^{-\Gamma t}$, and then pre- and post- multiply \eqref{eq:Sigmaminusinfty} by the integrating factors $\exp(-\Gamma^\top t)$ and $\exp(-\Gamma t)$ to obtain the differential equation
\[
\dot Y_t = - Y_t e^{\Gamma t}BB^\top e^{\Gamma^\top t} Y_t
\]
for $Y_t$. Continuing on, we define $Z_t=Y_t^{-1}$ and note that $Z_0=Y_0^{-1}=(\Sigma_0-\Sigma_\infty)^{-1}$. Since, $Z_t$ obeys $\dot Z_t= e^{\Gamma t}BB^\top e^{\Gamma^\top t}$,
\[
Z_t=Z_0+\int_0^t e^{\Gamma \tau}BB^\top e^{\Gamma^\top \tau} d\tau.
\]
Since $\Sigma_t=\Sigma_\infty + e^{\Gamma^\top t}Z_t^{-1}e^{\Gamma t}$, the claim follows.
\end{proof}

    The expression \eqref{eq:convergencerate} shows explicitly the dependence of the thermalization (i.e., convergence to equilibrium) of the quantum state on the detector parameters in $\Gamma$, which, in light of Proposition \ref{prop:prop1} can be expressed as
    \[
    \Gamma =\begin{bmatrix}
        -\eta_xk_xv_\infty^x & -m\omega^2-\eta_xk_x c_\infty^x\\
        \frac{1}{m}-\eta_pk_pc_\infty^p &   -\eta_pk_pv^p_\infty
    \end{bmatrix}.
    \]
    The analysis for the case when $\Sigma_0\not>\Sigma_\infty$ is more nuanced and will be detailed elsewhere.

\section{Steady-state regulation}\label{sec:IV}
We now investigate the effect of the detector parameters $\eta_x,\eta_p,k_x,k_p$ on the steady-state covariance $\Sigma_{\infty}$. 
The dynamics of the determinant of $\Sigma_t$ are derived and its behavior at steady-state is discussed in terms of the Robertson-Schr\"odinger uncertainty relation. The relation provides a fundamental limit on how much one can ``squeeze" a quantum state. 
In the Gaussian setting, the determinant fully characterizes the state \textit{purity}. The saturation of the uncertainty relation bound can only be achieved with ideal detectors.
It is shown that the purity of the system can be enhanced by introducing quantum correlations, and that the interval of achievable values for the purity at steady-state is solely characterized by the efficiencies of the detectors.
%
We conclude via a numerical example.

\subsection{The purity at steady-state}
The Robertson-Schr\"odinger uncertainty relation \cite{sen2014uncertainty,schrodinger1930heisenbergschen,sakurai2020modern}, applied to the observables $\hat{X}$ and $\hat{P}$, states that
\begin{align}
    v_t^xv_t^p-c_t^2\geq \left|\cfrac{1}{2i}\text{tr}\left(\rho_t[\hat{X},\hat{P}]\right) \right|^2 = \cfrac{\hbar^2}{4},\label{eq:RSUR}
\end{align}
for any quantum state $\rho_t$.
Thus, the determinant $\det(\Sigma_t)$ of a quantum state is bounded below by $\hbar^2/4$.

In the Gaussian setting, the dynamics of the determinant $d_t$ of $\Sigma_t$ are given by 
\begin{align}
    \dot{d}_t = \text{tr}(B\Sigma_t B)\left(\cfrac{\hbar^2}{4}\frac{\text{tr}(\chi B\Sigma_tB)}{\text{tr}(B\Sigma_t B)}-d_t\right),\label{eq:ddot}
\end{align}
and that at steady-state we have 
\begin{align}
    d_{\infty} = \cfrac{\hbar^2}{4}\cfrac{\text{tr}(\chi B\Sigma_{\infty}B)}{\text{tr}(B\Sigma_{\infty}B)} = \frac{\hbar^2}{4}\cfrac{k_xv_\infty^x+k_pv_\infty^p}{\eta_xk_xv_\infty^x+\eta_pk_pv_\infty^p},\label{eq:det}
\end{align}
 where 
\begin{align*}
    \chi = \begin{bmatrix}
        1/\eta_x&0\\
        0&1/\eta_p
    \end{bmatrix},~~0<\eta_x,\eta_p\leq 1.
\end{align*}
Thus, the uncertainty relation \eqref{eq:RSUR} saturates if and only if the detectors are ideal, i.e., $d_{\infty}=\hbar^2/4\iff\eta_x=\eta_p=1$.
Equation \eqref{eq:ddot} extends the one derived in \cite{karmakar2022stochastic} to account for inefficient detectors. Moreover, the determinant fully characterizes the \textit{purity} of the state, since,
\begin{align*}
    p &= (2\pi\hbar)\int_{\mathbb{R}^2}W^2_{\rho}(\mathbf{s})\text{d}\mathbf{s}\\
    &=\cfrac{\hbar}{2\pi\det(\Sigma)}\int_{\mathbb{R}^2}
 e^{-(\mathbf{s}-\text{\boldmath$\mu$})^\top\Sigma^{-1}(\mathbf{s}-\text{\boldmath$\mu$})}\text{d}\mathbf{s} = \cfrac{\hbar}{2\sqrt{\det(\Sigma)}}.
\end{align*}
Using (\ref{eq:ddot}-\ref{eq:det}), we have that
\begin{align*}
    \dot{p}_t=-\frac{p_t}{2}\left(\frac{\text{tr}(\chi B\Sigma_t B)}{\text{tr}(B\Sigma_t B)}p_t^2-1\right),
\end{align*}
and 
\begin{align}
    p_{\infty} = \sqrt{\cfrac{\text{tr}( B\Sigma_{\infty}B)}{\text{tr}(\chi B\Sigma_{\infty}B)}}.\label{eq:purity}
\end{align}
Evidently, $p_{\infty}=1$ if and only if $\eta_x=\eta_p=1$.

\subsection{The dependence of the steady-state purity on detector parameters}

We characterize herein the space of achievable values for $p_{\infty}$ as function of the detector parameters.  

\begin{prop}\label{eq:proppurity}Given the detector efficiencies $\eta_x,\eta_p>0$, then the following statements hold:
\begin{itemize}
    \item[i)] If $\eta_x=\eta_p=\eta$, then $ p_{\infty} = \sqrt{\eta}$,
    irrespective of the chosen strength parameters $k_x$ and $k_p$.
    \item[ii)] If  $\eta_x \neq \eta_p$, say $\eta_x > \eta_p$, then 
    \begin{align}\label{eq:interval}
        p_{\infty}\in(\sqrt{\eta_p},\sqrt{\eta_x}).
    \end{align}
    \end{itemize}
    Moreover, in case ii), any value in $(\sqrt{\eta_p},\sqrt{\eta_x})$ can be achieved by a choice of $q:=k_xk_p>0$ and $s=k_x/k_p>0$ satisfying
   \begin{align}\label{eq:relation}
        \left(1-\cfrac{p_\infty^2}{\eta_x}\right)\gamma(q,s) + \left(1-\cfrac{p_\infty^2}{\eta_p}\right)\delta(q,s) = 0,
    \end{align}
    where $\gamma,\delta$ in \eqref{eq:greek} are seen as functions of $q$ and $s$.
    \end{prop}
    \begin{proof} If the detector efficiencies are equal to $\eta$, then $\chi = I/\eta$ and \eqref{eq:purity} gives the result. If the detector efficiencies are not equal, say $\eta_x>\eta_p$, then from our closed form solution in \eqref{eq:prop1} the measurement strengths must be chosen so that 
    \begin{align}
        p_{\infty}^2 = \cfrac{\text{tr}( B\Sigma_{\infty}B)}{\text{tr}(\chi B\Sigma_{\infty}B)} = \cfrac{\eta_xk_xv_\infty^x+\eta_pk_pv_\infty^p}{k_xv_\infty^x+k_pv_\infty^p},\label{eq:proofth2}
    \end{align}
    where $v_{\infty}^x$ is as defined in  \eqref{eq:explicit}. Plugging the expressions for $v_\infty^x$ and $v_\infty^p$ in \eqref{eq:proofth2} yields
    \begin{align*}
        \left(1-\cfrac{p_\infty^2}{\eta_x}\right)\gamma + \left(1-\cfrac{p_\infty^2}{\eta_p}\right)\delta = 0.
    \end{align*}
Finally, the surjectivity of the map $(q,s)\mapsto p_{\infty}\in(\sqrt{\eta_p},\sqrt{\eta_p})$ can be established by re-arranging \eqref{eq:relation} into
   \begin{align*}
      p_{\infty}(q,s)=\sqrt{\cfrac{\gamma(q,s)+\delta(q,s)}{\frac{\gamma(q,s)}{\eta_x}+\frac{\delta(q,s)}{\eta_p}}},
  \end{align*}
   and observing that for vanishingly small $q$, the whole range $(\sqrt{\eta_p},\sqrt{\eta_p})$ can be achieved for $s\in(0,\infty)$ .
   \end{proof}
\begin{figure}
\centering
  \includegraphics[width=.7\linewidth]{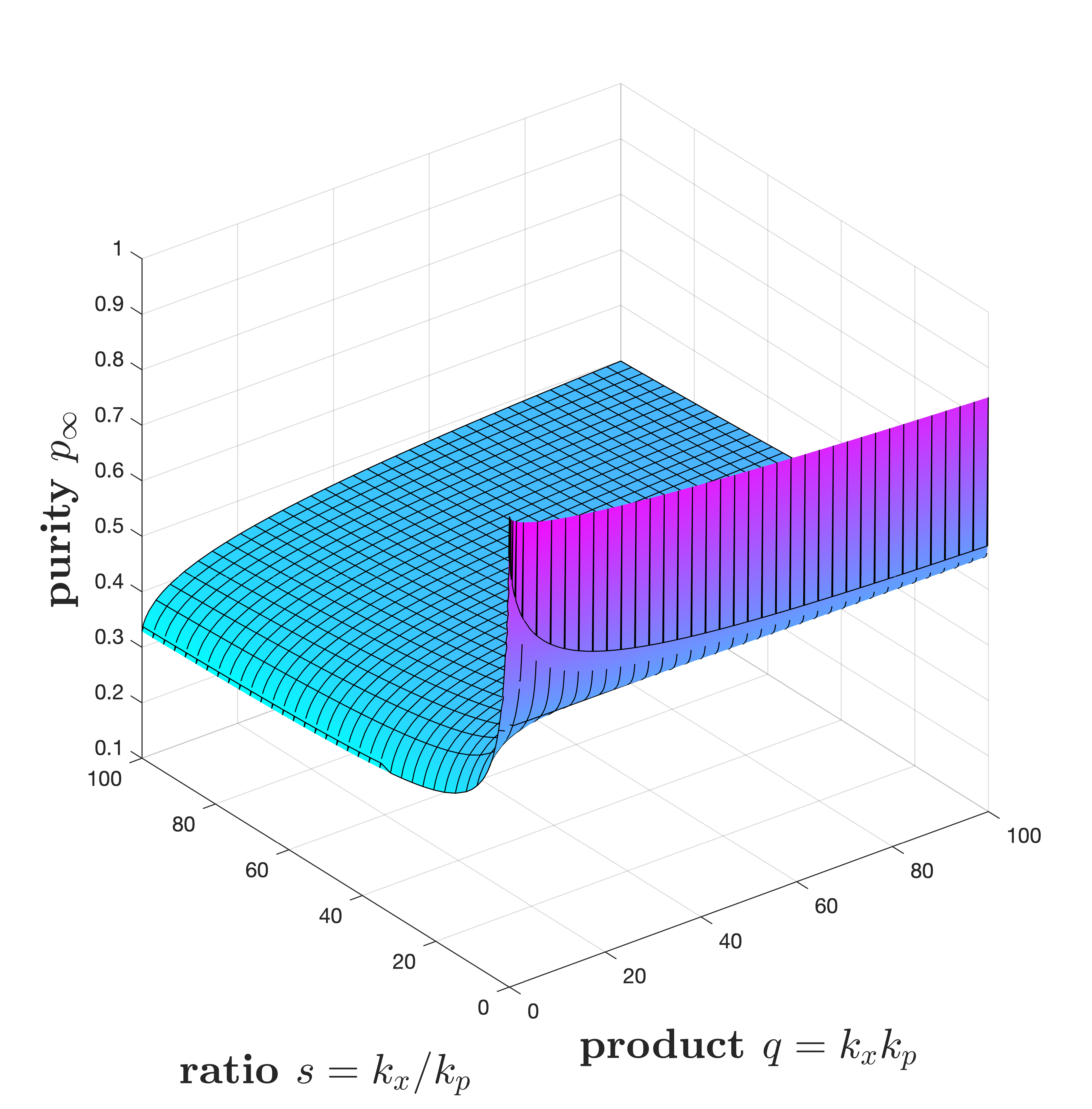}
  \caption{The purity $p$ at steady-state versus the product and ratio $q=k_xk_p$ and $s=k_x/k_p$, respectively. We take $m=\omega =1$, $\hbar = 2$, $\eta_x = 0.1$, and $\eta_p$ = 0.9.}
  \label{fig:first}
  \end{figure}
  \begin{figure}
\centering
  \includegraphics[width=.7\linewidth]{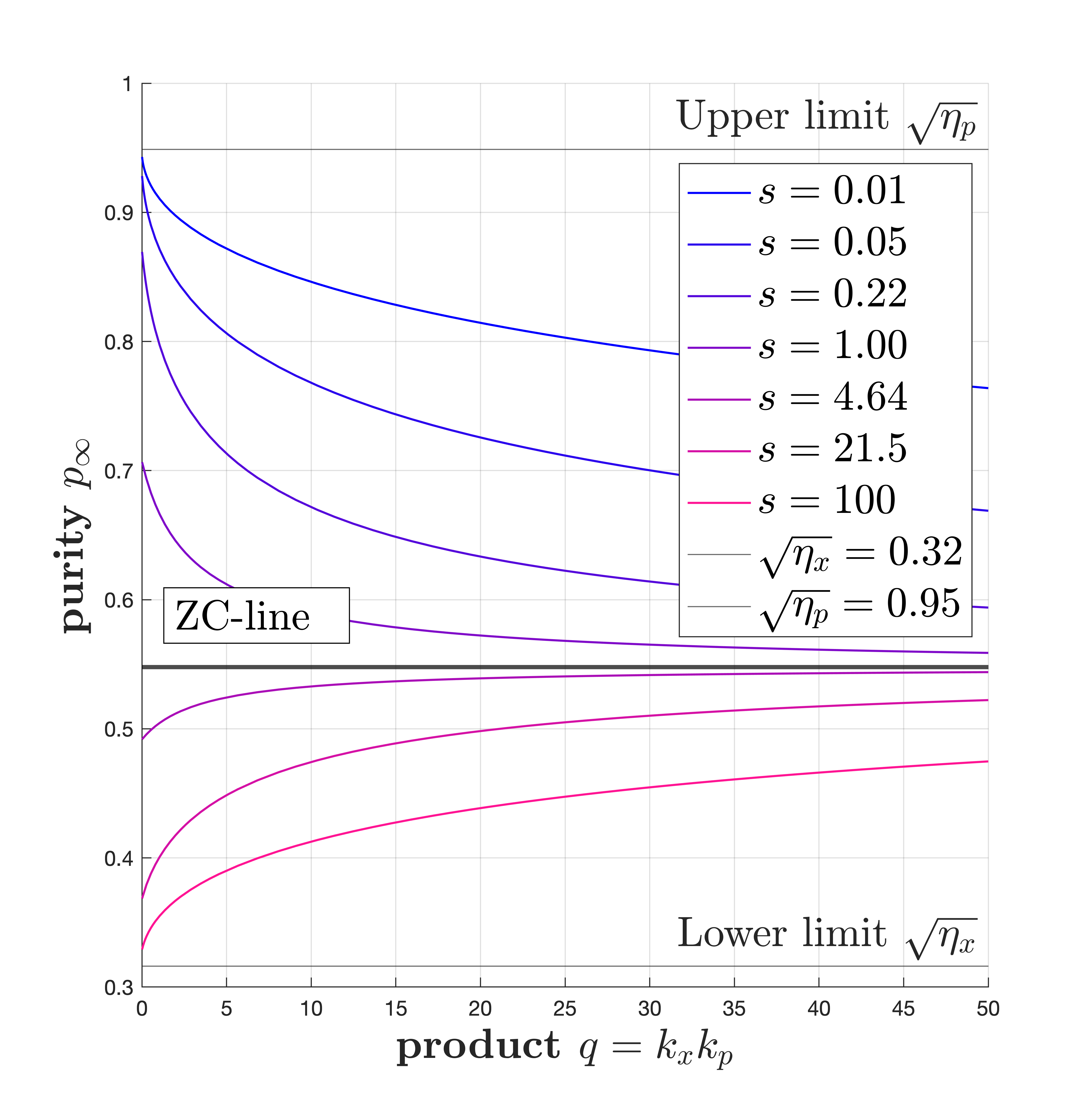}
  \caption{The purity $p$ at steady-state versus the product $q=k_xk_p$ for fixed values of the ratio $s=k_x/k_p$. We take $m=\omega =1$, $\hbar = 2$, $\eta_x = 0.1$, and $\eta_p$ = 0.9.}
  \label{fig:second}
\end{figure}
  \begin{figure}
  \centering
  \includegraphics[width=.7\linewidth]{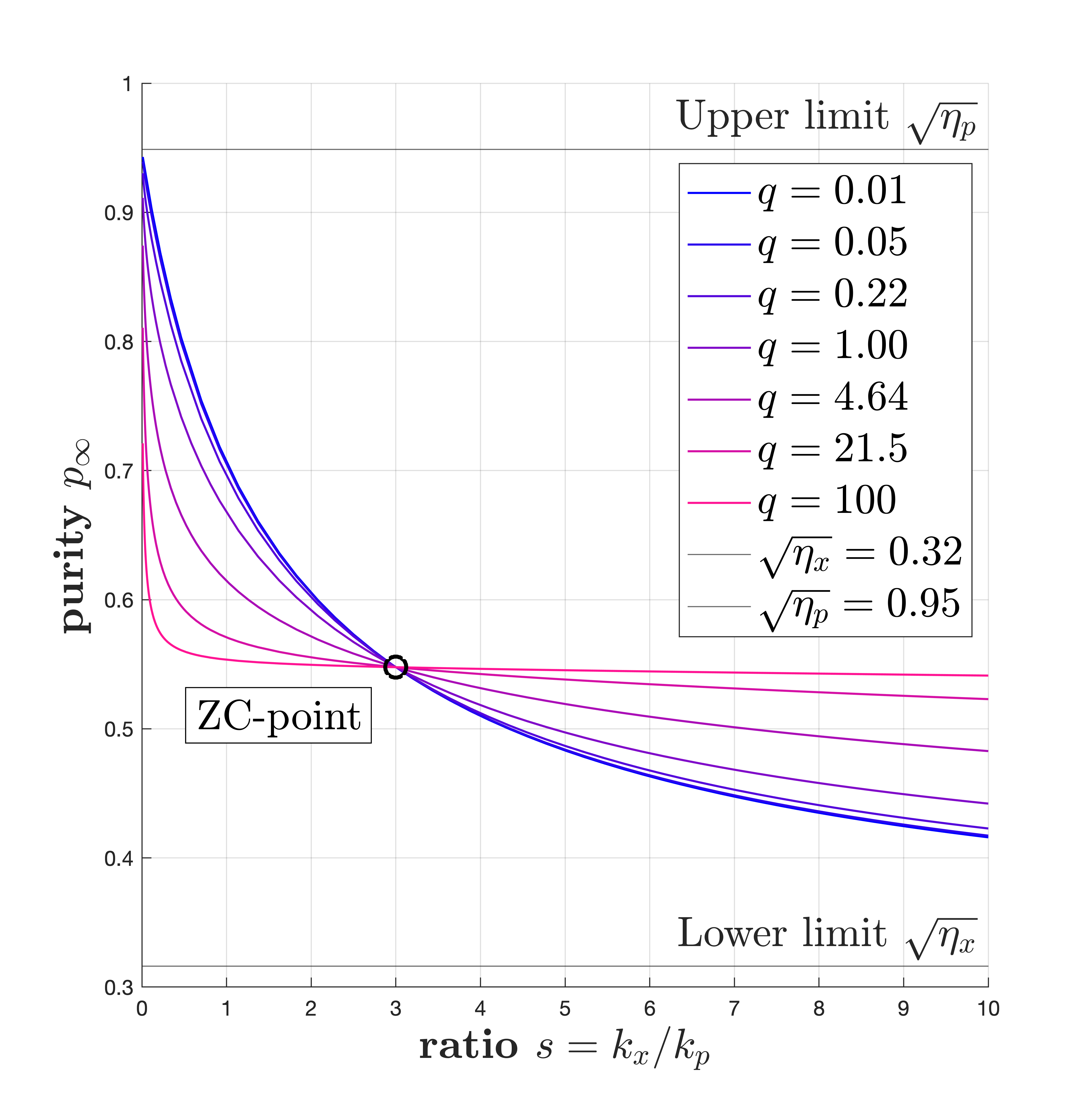}
  \caption{The purity $p$ at steady-state versus the ratio $s=k_x/k_p$ for fixed values of the product $q=k_xk_p$. We take $m=\omega =1$, $\hbar = 2$, $\eta_x = 0.1$, and $\eta_p$ = 0.9. }
  \label{fig:third}
\end{figure}

\subsection{Zero-Correlation point}
 A particular feature of the dependence of $p_\infty$ on the detector parameters is observed when the value of the correlation coefficient ($c_{\infty}/\sqrt{v_\infty^xv_\infty^p}$) is zero. This is achieved for a unique value of the ratio $s$, irrespective of $q$, and the steady-state purity is given by $(\eta_x\eta_p)^{1/4}$. We refer to this ratio, along with its corresponding $p_{\infty}$ as the \textit{Zero-Correlation (ZC) point}. 
\begin{prop}\label{prop0}
    Given the detector efficiencies $\eta_x,\eta_p\in(0,1]$, and measurement strengths $k_x,k_p>0$, a zero steady-state correlation is achieved if and only if 
    \begin{align}
        \cfrac{k_x}{k_p}=m^2\omega^2\sqrt{\cfrac{\eta_p}{\eta_x}}.\label{eq:condittion}
    \end{align}
\end{prop}

\vsp
\begin{proof}
    Setting $c_{\infty}=0$ in \eqref{eq:explicit} implies that $\alpha\gamma+\beta\delta = 0$. By expanding and simplifying this identity, the result follows.
\end{proof}
    At the zero-correlation point, $\Sigma_\infty$ is 
    \begin{align*}
       \cfrac{\hbar}{2(\eta_x\eta_p)^{1/4}} \begin{bmatrix}
            (m\omega)^{-1}&0\\
            0&m\omega\end{bmatrix},
    \end{align*}
    and the purity at steady-steady is given by the geometric mean of $\sqrt{\eta_x}$ and $\sqrt{\eta_p}$. Thus, the ZC point is the pair \[
    s= m^2\omega^2\sqrt{\cfrac{\eta_p}{\eta_x}}, \mbox{ and } p_\infty(s)=(\eta_x\eta_p)^{1/4},
    \]
    and the uncertainty relation is saturated at $\hbar/(2(\eta_x\eta_p)^{1/4}).$ 
    
    The above discussion highlights that the ZC steady-state is not optimal in the sense of maximizing purity, since by \ref{eq:proppurity}, the entire interval $(\sqrt{\eta_x},\sqrt{\eta_p})$ can be reached with suitable detector strengths. Thus, the purity $p_{\infty}$ can only be improved beyond $(\eta_x\eta_p)^{1/4}$ by introducing non-zero correlations at steady-state. This is one of the main observations in this work.

  Figures \ref{fig:first}-\ref{fig:third}, illustrate a numerical example showcasing the behavior of the steady-state purity $p_{\infty}$ as a function of $q$ and $s$ for given values of $\eta_x$ and $\eta_p$. It is interesting to note that the zero-correlation threshold splits the behavior of $p_{\infty}$ into two regimes. When $s$ is more than the ZC threshold, the purity is bounded above by this threshold and can be improved for large values of $q$. When $s$ is less than the ZC threshold,  the purity if bounded below by this threshold and can be improved for small values of $q$. It is interesting to observe that to mazimize the purity at steady-state, both $s$ and $q$ must be made small, with the $k_p$ strength significantly larger than the $k_x$ strength.  

  In our final figure, Fig.~\ref{fig:four}, we illustrate steady-state purity regulation of a zero-correlation Gaussian state by slowly varying the detector strengths. Specifically, we vary the control parameters $k_x$ and $k_p$ at a much slower rate than the thermalization time scale of the system (i.e., the rate at which the Riccati equation converges), a standard assumption in experiments \cite{karmakar2022stochastic}. The figure displays iso-probability levels (``confidence regions'') for $21$ instances of the regulation parameters. It is seen that squeezing the ZC Gaussian state introduces correlations.

\section{CONCLUSIONS}

We considered the control problem to regulate the purity of a Gaussian quantum state in a harmonic potential via continuous measurements, where the control parameters are given by the strengths of non-ideal detectors monitoring position and momentum. A future direction of great interest is to consider transients and cyclic operation of quantum system, regulating the Hamiltonian as well as using monitoring protocols that adjust detector parameters, for the purpose of quantifying entropy production \cite{landi}, and eventually, work production and power in quantum engines \cite{park2013heat}, echoing a stochastic thermodynamic framework akin to \cite{olga,olga1}.

\begin{figure}
    \centering
    \includegraphics[width=\linewidth]{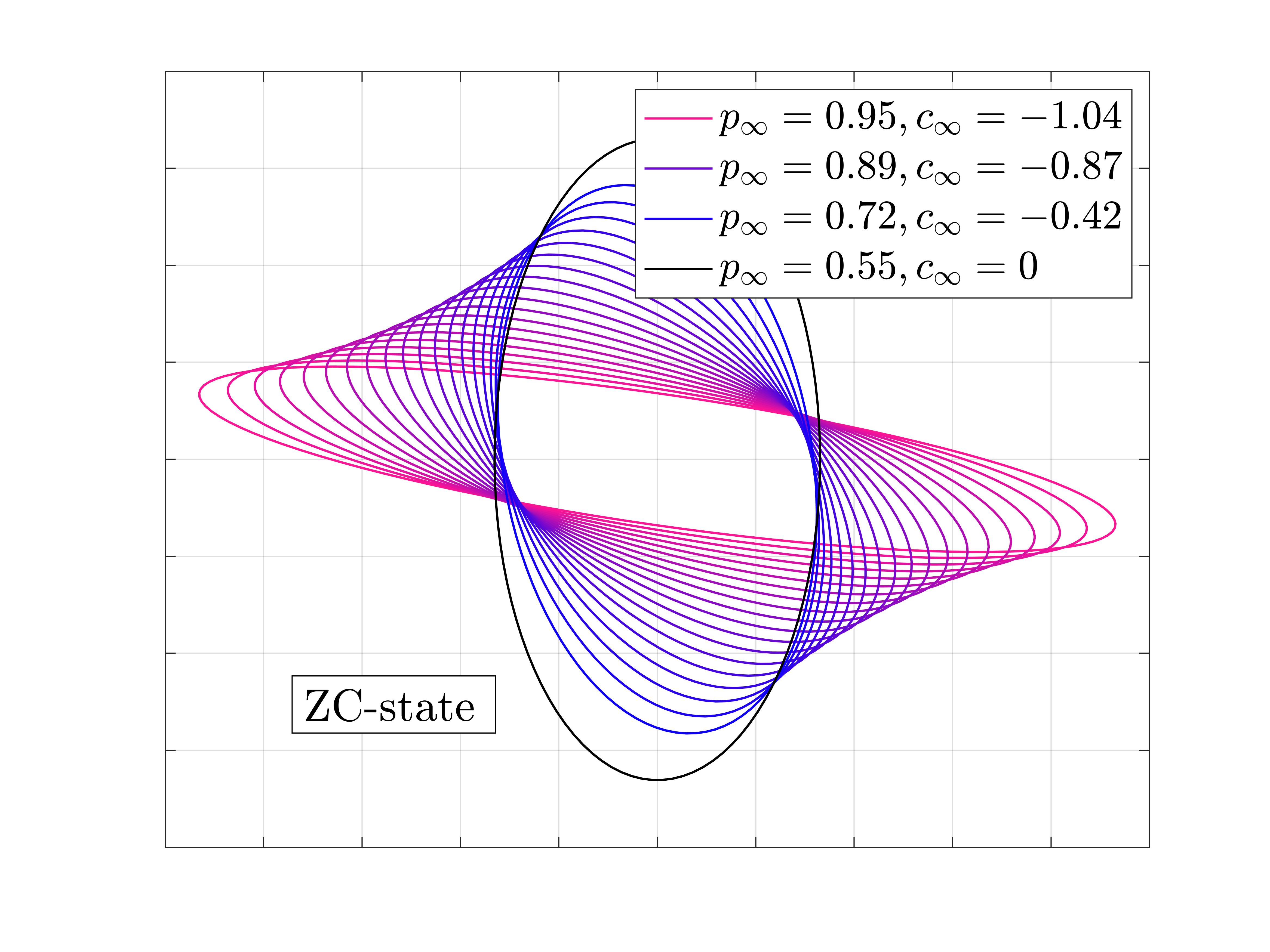}
    \caption{Purification/squeezing of a zero-correlation initial Gaussian state (shown in black) by slowly varying $s$ from $3$ to $0.0001$ with $q$ fixed at $1$. Starting from the zero correlation purity $\sqrt{\sqrt{\eta_x}\sqrt{\eta_p}}\approx 0.55$ where $s = 3$, the maximal value $\sqrt{\eta_p}=\sqrt{0.9}\approx 0.95$ of $p_{\infty}$ is approached as $s$ decrease, while the correlation goes from $0$ to $-1.04$.}
    \label{fig:four}
    \vspace{-3mm}
\end{figure}

\bibliographystyle{IEEEtran}
\bibliography{References.bib}
\end{document}